\documentclass[technote,onecolumn,11pt,a4paper]{IEEEtran}

\usepackage[cmex10]{amsmath} 
\usepackage{amssymb,amsthm}
\usepackage{graphicx,gastex}
\usepackage{subfig}
\usepackage{multirow}
\usepackage{cite}
\usepackage{xcolor}
\usepackage{tikz} 
\pgfrealjobname{randomaccess_arxiv}

\usetikzlibrary{decorations.pathreplacing}
\usetikzlibrary{positioning}
\usetikzlibrary{calc}

\newcommand{\fm}[1]{\scriptsize\mbox{\ensuremath{#1}}}
\newcommand{\ft}[1]{\scriptsize #1}

\newtheorem{theorem}{Theorem}

\newtheorem{definition}{Definition}

\newcommand{\ie}{{\it i.e.},\ }

\newcommand{\cf}{{\it cf.} }

\usepackage{bbm}

\newcommand{\zz}{\mathbf{0}}
\newcommand{\bb}{\mathbf{b}}
\renewcommand{\gg}{\mathbf{g}}
\renewcommand{\SS}{\mathbf{S}}

\newcommand{\MM}{\mathbf{M}}
\newcommand{\GG}{\mathbf{G}}
\renewcommand{\AA}{\mathbf{A}}
\newcommand{\RR}{\ensuremath{\mathbb{R}}}

\newcommand{\EE}{\ensuremath{\mathbb{E}}}
\newcommand{\FF}{\ensuremath{\mathbb{F}}}

\newcommand{\Pe}{\ensuremath{P_e}}

\DeclareMathOperator{\diag}{diag}

\begin{document}

\title{Random Access with Physical-layer Network Coding}

\author{
\IEEEauthorblockN{
Jasper Goseling\IEEEauthorrefmark{1}\IEEEauthorrefmark{4},
Michael Gastpar\IEEEauthorrefmark{2}\IEEEauthorrefmark{3}\IEEEauthorrefmark{4} and
Jos H. Weber\IEEEauthorrefmark{4} 
}\\
\IEEEauthorblockA{
 \IEEEauthorrefmark{1}
 Stochastic Operations Research,\\
 University of Twente, The Netherlands\\
}
\IEEEauthorblockA{
 \IEEEauthorrefmark{2}
 Laboratory for Information in Networked Systems,\\
 Ecole Polytechnique F\'ed\'erale de Lausanne, Switzerland\\
}
\IEEEauthorblockA{
 \IEEEauthorrefmark{3}
Dept. of EECS, University of California, Berkeley\\
}
\IEEEauthorblockA{
 \IEEEauthorrefmark{4}
 Delft University of Technology, The Netherlands\\
 j.goseling@utwente.nl, michael.gastpar@epfl.ch, j.h.weber@tudelft.nl
}
}

\maketitle

%
%
%
\begin{abstract}
Leveraging recent progress in physical-layer network coding we propose a  
new approach to random access: When packets collide, it is possible to  
recover a linear combination of the packets at the receiver. Over many  
rounds of transmission, the receiver can thus obtain many linear  
combinations and eventually recover all original packets. This is by  
contrast to slotted ALOHA where packet collisions lead to complete  
erasures. The throughput of the proposed strategy is derived and shown  
to be significantly superior to the best known strategies, including  
multipacket reception.
\end{abstract}

%
%
%
\section{Introduction}

Consider a wireless network with multiple users transmitting to a single receiver using a random access mechanism. Users are in one of two states, active or inactive, and do not have knowledge of the states of other users. The receiver has complete knowledge of the states of all users. The network is operated in rounds.
In each round each user chooses his state at random independently of the state of other rounds and independent of the other users. The problem studied in the current paper is how to, given the constraints of the random access mechanism, organize the physical layer. This involves, for example, selecting a type of forward-error correcting code and the rate at which it is used. The \emph{contribution} of the current paper is a new approach to random access, based on \emph{physical-layer network coding}. 

%
%
\begin{figure}
\hfill
\subfloat[Slotted ALOHA\label{fig:introaloha}]{
\beginpgfgraphicnamed{pgftimea}
\begin{tikzpicture}[scale=0.8]

\newcommand{\drawslot}[3]{
\begin{scope}[xshift=#1,yshift=#2]
\draw[rounded corners=1pt,color=black,fill=black!10!white] (.1,.1) rectangle (1.9,1.3);
\foreach \x in {.2,.3,.4,...,1.8}
	{
	\draw[black!20!white] (\x,.1) -- (\x,1.3);
	}
\foreach \y in {.2,.3,.4,...,1.2}
  {
  \draw[black!20!white] (.1,\y) -- (1.9,\y);
  }
\draw[rounded corners=1pt,color=black] (.1,.1) rectangle (1.9,1.3);
\node[rectangle,inner sep=0] at (1,.7) {\fm{#3}};
\end{scope}
}

\draw[-latex] (0,-.1) -- (12.5,-.1);
\foreach \slot in {0,2,4,...,12}
	{
	\draw (\slot,-.2) -- (\slot,0);
	}

\node[anchor=east] at (-.1,0.7) {\ft{User $2$}};
\node[anchor=east] at (-.1,2.1) {\ft{User $1$}};
\node[anchor=east] at (-.1,-.9) {\ft{Receiver}};

\drawslot{2cm}{0cm}{C}
\drawslot{4cm}{0cm}{C}
\drawslot{10cm}{0cm}{D}

\drawslot{0cm}{1.4cm}{A}
\drawslot{2cm}{1.4cm}{B}
\drawslot{8cm}{1.4cm}{B}

\drawslot{4cm}{-1.6cm}{C}
\drawslot{10cm}{-1.6cm}{D}

\drawslot{0cm}{-1.6cm}{A}
\drawslot{8cm}{-1.6cm}{B}


\end{tikzpicture}
\endpgfgraphicnamed
}
\hfill{}

\hfill
\subfloat[Multipacket reception (MPR)\label{fig:intrompr}]{
\beginpgfgraphicnamed{pgftimeb}
\begin{tikzpicture}[scale=0.8]

\newcommand{\drawslot}[3]{
\begin{scope}[xshift=#1,yshift=#2]
\draw[rounded corners=1pt,color=black,fill=black!10!white] (.1,.1) rectangle (1.9,.6);
\foreach \x in {.2,.3,.4,...,1.8}
	{
	\draw[black!20!white] (\x,.1) -- (\x,.6);
	}
\foreach \y in {.2,.3,.4,.5}
  {
  \draw[black!20!white] (.1,\y) -- (1.9,\y);
  }
\draw[rounded corners=1pt,color=black] (.1,.1) rectangle (1.9,.6);
\node[rectangle,inner sep=0] at (1,.35) {\fm{#3}};
\end{scope}
}

\draw[-latex] (0,-.1) -- (12.5,-.1);
\foreach \slot in {0,2,4,...,12}
	{
	\draw (\slot,-.2) -- (\slot,0);
	}

\node[anchor=east] at (-.1,.35) {\ft{User $2$}};
\node[anchor=east] at (-.1,1) {\ft{User $1$}};
\node[anchor=east] at (-.1,-.55) {\ft{Receiver}};

\drawslot{2cm}{0cm}{C}
\drawslot{4cm}{0cm}{D}

\drawslot{0cm}{.7cm}{A}
\drawslot{2cm}{.7cm}{B}

\drawslot{2cm}{-1.6cm}{C}
\drawslot{4cm}{-.9cm}{D}
\drawslot{0cm}{-.9cm}{A}
\drawslot{2cm}{-.9cm}{B}


\end{tikzpicture}
\endpgfgraphicnamed
}
\hfill{}

\hfill
\subfloat[Physical-layer network coding (PLNC)\label{fig:introplnc}]{
\beginpgfgraphicnamed{pgftimec}
\begin{tikzpicture}[scale=.8]

\newcommand{\drawslot}[3]{
\begin{scope}[xshift=#1,yshift=#2]
\draw[rounded corners=1pt,color=black,fill=black!10!white] (.1,.1) rectangle (1.9,1.3);
\foreach \x in {.2,.3,.4,...,1.8}
	{
	\draw[black!20!white] (\x,.1) -- (\x,1.3);
	}
\foreach \y in {.2,.3,.4,...,1.2}
  {
  \draw[black!20!white] (.1,\y) -- (1.9,\y);
  }
\draw[rounded corners=1pt,color=black] (.1,.1) rectangle (1.9,1.3);
\node[rectangle,inner sep=0] at (1,.7) {\fm{#3}};
\end{scope}
}

\newcommand{\plus}{
\mathbin{\tikz[scale=.8]{\draw[black,ultra thin] (-2pt,0pt) -- (2pt,0pt); \draw[black,ultra thin] (0pt,-2pt) -- (0pt,2pt);}}
}

\draw[-latex] (0,-.1) -- (12.5,-.1);
\foreach \slot in {0,2,4,...,12}
	{
	\draw (\slot,-.2) -- (\slot,0);
	}

\node[anchor=east] at (-.1,.7) {\ft{User $2$}};
\node[anchor=east] at (-.1,2.1) {\ft{User $1$}};
\node[anchor=east] at (-.1,-.9) {\ft{Receiver}};

\drawslot{2cm}{0cm}{gC\!\plus\!hD}
\drawslot{4cm}{0cm}{kC\!\plus\!lD}

\drawslot{0cm}{1.4cm}{aA\!\plus\!bB}
\drawslot{2cm}{1.4cm}{cA\!\plus\!dB}
\drawslot{8cm}{1.4cm}{eA\!\plus\!fB}

\drawslot{4cm}{-1.6cm}{kC\!\plus\!lD}

\drawslot{0cm}{-1.6cm}{aA\!\plus\!bB}
\drawslot{2cm}{-1.6cm}{}
\node at (3,-.67) {\fm{cA\!\plus\!dB}};
\node at (3,-1.13) {\fm{\plus gC\!\plus\!hD}};
\drawslot{8cm}{-1.6cm}{eA\!\plus\!fB}


\end{tikzpicture}
\endpgfgraphicnamed
}
\hfill{}
\caption{Illustration of various approaches to random access. The height of the packets reflects the amount of information contained in the packet.}
\end{figure}
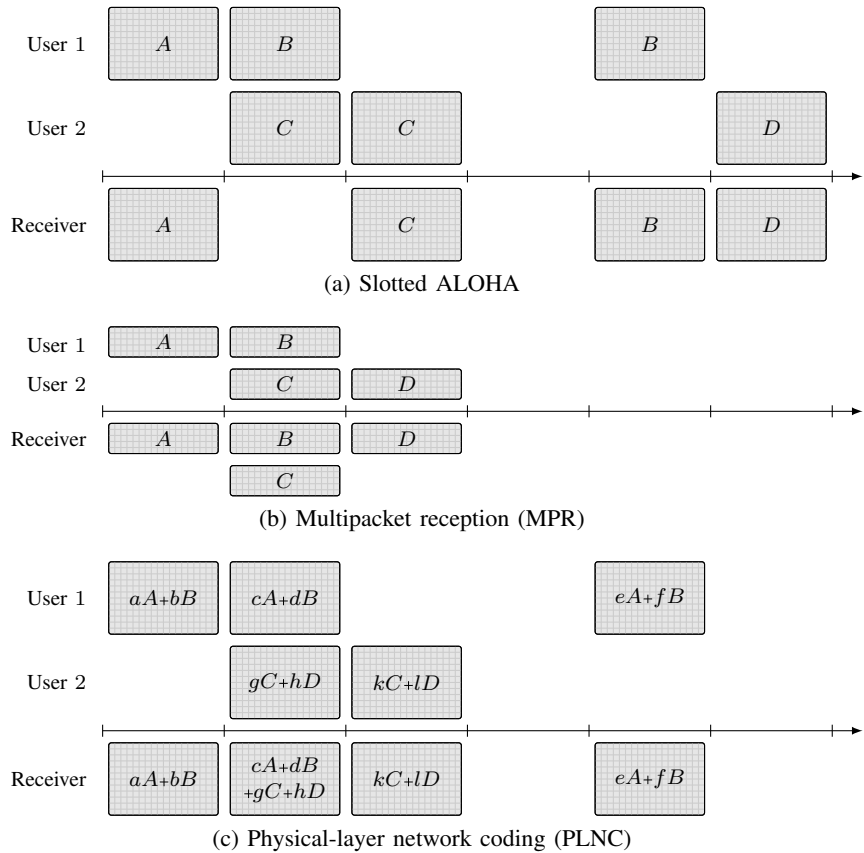

%
%

The most thoroughly studied approach to random access is slotted ALOHA~\cite{abramson1970aloha}, \cf~\cite{bertsekasgallager}. In ALOHA, if more than one user is active, a packet collision occurs and the receiver does not obtain any information about the transmitted packets. This is illustrated in~Figure~\ref{fig:introaloha}, in which different rounds are represented along the horizontal axis. Packets transmitted by the users are depicted above the axis, the packets below the axis represent the information obtained by the receiver. In order to compare various strategies, we let the height of a packet reflect the amount of information contained in the packet, \ie it reflects the rate of the underlying forward-error correcting code.


It is well known, see for instance~\cite{elgamalkim}, that if the communication rates of the users are chosen carefully the receiver can decode all, or a subset of, the packets of the active users. This is often referred to as \emph{multipacket reception} and its use for random access was considered in~\cite{ghez88mpr}. It is illustrated in Figure~\ref{fig:intrompr}. As reflected in the figure, multipacket reception requires the rate to be adjusted. More recent results on multipacket reception for random access are given in~\cite{minero2012random, medard2004capacity, minero2009mpr}). One of the aspects studied in~\cite{minero2012random} is the tradeoff between the rate of the code and the maximum number of packets that can be decoded.

The \emph{strategy proposed in the current paper} is based on another way of dealing with collisions. Instead of trying to decode any of the packets transmitted by the users, the receiver attempts to decode a linear combination of these packets. This is known as \emph{reliable physical-layer network coding}~\cite{nazer2011procieee}. In the proposed strategy, the users transmit in each time slot a packet that is itself a linear combination of the messages intended for the receiver. After obtaining a sufficient number of linear combinations the receiver can retrieve the original messages. The physical-layer network coding strategy is illustrated in Figure~\ref{fig:introplnc}. For the example in the figure, the information obtained by the receiver can be represented as
\begin{equation} \label{eq:intromatrix}
\begin{bmatrix}
a & b & 0 & 0 \\
c & d & g & h \\
0 & 0 & k & l \\
e & f & 0 & 0 
\end{bmatrix}
\begin{bmatrix}
 A \\ B \\ C \\ D
\end{bmatrix}
=
\begin{bmatrix}
b_1 \\ b_2 \\ b_3 \\ b_4
\end{bmatrix},
\end{equation}
where $A,\cdots,D$ denote the messages and $b_1,\dots,b_4$ the information obtained by the receiver at the channel output. It will be shown that our strategy significantly outperforms existing strategies.

The concept of physical-layer network coding is studied in, for instance~\cite{nazer07comp,narayanan2007joint,popovski2007physical,rankov2006achievable,zhang06hot,katti07analog}. See~\cite{nazer2011procieee} for a survey of recent results. The canonical example demonstrating the use of decoding linear combinations of packets is that of two-way relaying. By applying network coding~\cite{yeung1999distributed,wu05ciss}, the relay retransmits a linear combination of packets. As a consequence, the relay might as well obtain this linear combination directly. In~\cite{nazer07comp} and~\cite{narayanan2007joint} the aim is to obtain these linear combinations reliably. In contrast, in~\cite{popovski2007physical,rankov2006achievable,zhang06hot,katti07analog} the relay is satisfied with a noisy version of these linear combinations, which are then retransmitted. Our strategy is based on \emph{reliable} physical layer network coding.


Slotted ALOHA has received a lot of attention in the literature, \cf~\cite{bertsekasgallager} and references therein. Many aspects are well understood and improvements on slotted ALOHA  have been suggested along a multitude of directions. There is, for example, a good understanding of throughput, stability, delay and the impact of selfish users~\cite{mackenzie2003stability}. Improvements on ALOHA that have been suggested include the use of collision resolution protocols and to introduce cooperation between users~\cite{karamchandani2011cooperation}. In the current paper we restrict our attention to studying the throughput. Also, we will compare the proposed strategy with other strategies that do not employ feedback from the receiver and do not allow for cooperation.

Our model and our definition of throughput, to be defined precisely in Section~\ref{sec:model}, are rooted in information theory. This has the advantage that it allows to not only analyze the performance of the proposed strategy, but to also compare it with an upper bound on the throughput that must hold for all strategies. The model at hand is that of a  time-varying multiple access channel with partial state information at the transmitters and full state information at the decoder~\cite{elgamalkim}. While there exist some capacity characterizations for this class of models~\cite{das2002capacities}, no methods are known to evaluate these characterizations and obtain numerical values. The additional structure and properties of the random access model allow us to obtain upper bounds on the throughput that can be numerically evaluated. Other examples exists of special cases in which more explicit capacity characterizations are known, see for instance~\cite{cemal2005multiple}.



The outline of the remainder of this paper is as follows. In Section~\ref{sec:model} we define the model. Reliable physical-layer network coding is briefly introduced in Section~\ref{sec:compfw}. The main contributions of the current paper, a description of our approach and an analysis of the resulting throughput are given in Section~\ref{sec:strategy}. Section~\ref{sec:comparison} provides an analysis of the throughput of some other approaches as well as an upper bound on the achievable throughput. A numerical evaluation of results is given in~\ref{sec:results}. Finally, in Section~\ref{sec:discussion} we conclude with a discussion of the results that are presented in the current paper and an outlook on future work.

%
%
%
\section{Reliable Physical-layer Network Coding} \label{sec:compfw}

One key ingredient of the strategy proposed in this paper is the technique of so-called ``computation codes''
and the recently introduced cooperative communication technique entitled ``compute-and-forward''.
We provide an informal overview of this technique here, but refer the reader to~\cite{nazer07comp,nazer11compforw,nazer2011procieee}
for technical details. A related approach was taken for the special case of the two-way relay channel
in \cite{narayanan2007joint}.

To set the stage for the technique, we consider two transmitters, each having an independent message.
Moreover, we think of the two messages as being represented as strings over an appropriately chosen
finite field ${\mathbb F}_q.$
That is, we denote the message of transmitter 1 as
\begin{align}
  M_1 & = ( M_1(1), M_1(2), \ldots, M_1(L_c) ),
\end{align}
and the message of transmitter 2 as
\begin{align}
  M_2 & = ( M_2(1), M_2(2), \ldots, M_2(L_c) ),
\end{align}
where $M_i(j)$ is an element of ${\mathbb F}_q.$
Thus, at each transmitter, there are $q^L$ different possible messages.
Each transmitter can encode its message into a string of $B$ real numbers
satisfying an average power constraint $P.$
In line with standard information-theoretic terminology, we define the {\em rate} of the resulting
two codes, which is the same for both transmitters, by
\begin{equation}
 R =  R_1 = R_2  =  \frac{L_c \log_2 q}{ B}  \mbox{ bits per channel use.}  \label{Eqn-compfw-rate}
\end{equation}

The real-valued strings of length $B,$ which we denote simply as $X_1$ and $X_2,$ are then transmitted
element-wise across the standard AWGN multiple-access channel, whose channel output is given by
\begin{equation}
  Y  =  X_1 + X_2 + Z,
\end{equation}
where $Z$ is additive white Gaussian noise with unit variance.

The decoder, upon observing the real-valued string $Y$ of length $B,$ is asked to provide
an estimate sequence $(\hat{M}(1), \hat{M}(2), \ldots, \hat{M}(L_c))$
in such a way as to minimize the probability of the event
\begin{equation}
  (\hat{M}(1), \hat{M}(2), \ldots, \hat{M}(L_c)) 
  \not=   ( M_1(1) + M_2(1), \ldots, M_1(L_c) + M_2(L_c) ).
\end{equation}
In this sense, the receiver recovers a function (namely, the sum) of the original messages, which is why this approach
is referred to as computation coding.
We refer to the rate $R_1 = R_2$ as the {\em computation rate.}
We then ask the standard information-theoretic key question:
How large one can make the computation rate such that 
the probability of the above event can be made arbitrarily small as $B$ increases?

The best currently known scheme~\cite{nazer07comp,nazer11compforw},
and the one we will employ in the present paper,
involves using one and the same code at both encoders, namely,
\begin{eqnarray}
  X_1 = F(M_1) & \mbox{and} &  X_2 = F(M_2)  , \label{Eq-compcodeH}
\end{eqnarray}
and we refer to $F(\cdot)$ as the {\em computation code.}
It can be shown that there exists such a computation code
as long as the computation rate satisfies $R \le \frac{1}{2} \log_2 ( \frac{1}{2} + P).$
Interestingly, this result cannot be established using the standard information-theoretic random coding
arguments. Rather, the proof employs random {\em lattice} code constructions.
In straightforward extension of the above, we can look at $K$ transmitters and ask the receiver
to recover the sum of all $K$ message strings.
Then, we have the following statement:
\begin{theorem}[\cite{nazer11compforw}, Thm. 2]\label{th:compfw}
For the standard AWGN multiple-access channel, the following computation rate is achievable:
\begin{eqnarray}
  R & = & \frac{1}{2} \log_2 \left( \frac{1}{K} + P \right).
\end{eqnarray}
\end{theorem}

%
%
%
\section{Model and Problem Formulation} \label{sec:model}
\begin{figure}
\hfill
\beginpgfgraphicnamed{pgfmodel}
\begin{tikzpicture}[scale=.5,yscale=-1]
\tikzstyle{boxa}=[draw, rounded corners,minimum width=12mm,minimum height=8mm];
\tikzstyle{boxb}=[draw, circle];
\tikzstyle{->}=[-latex];
\node[boxa,anchor=east] (enca) at (-4,-2) {\ft{Enc $1$}}; 
\node[boxa,anchor=east] (encb) at (-4,2) {\ft{Enc $2$}}; 
\node[boxb] (statea) at (-1,-2) {\fm{\times}};
\node[boxb] (stateb) at (-1,2) {\fm{\times}};
\node[boxb] (noise) at (2,0) {\fm{+}};
\node[boxa,anchor=west] (dec) at (5,0) {\ft{Dec}};

\draw[->] ($ (enca.west) - (1cm,0) $) -- (enca);
\draw[->] ($ (enca.north) - (0,1cm) $) -- node[at start,above] {\fm{S_1}} (enca);
\draw[->] (enca) -- node[above,near start] {\fm{X_1}} (statea);
\draw[->] ($ (statea.north) - (0,1cm) $) -- node[at start,above] {\fm{S_1\in\{0,1\}}} (statea);
\draw[->] (statea) to (noise);

\draw[->] ($ (encb.west) - (1cm,0) $) -- (encb);
\draw[->] ($ (encb.south) + (0,1cm) $) -- node[at start,below] {\fm{S_2}} (encb);
\draw[->] (encb) to node[above,near start] {\fm{X_2}} (stateb);
\draw[->] ($ (stateb.south) + (0,1cm) $) -- node[at start,below] {\fm{S_2\in\{0,1\}}} (stateb);

\draw[->] (stateb) to (noise);
\draw[->] ($ (noise.north) - (0,1cm) $) -- node[at start,above] {\fm{Z}} (noise); 
\draw[->] (noise) to node[above,near end] {\fm{Y}} (dec);
\draw[->] ($ (dec.north) - (0,1cm) $) -- node[at start,above] {\fm{S_1, S_2}} (dec); 
\draw[->] (dec) -- ($ (dec.east) + (1cm,0) $);
\end{tikzpicture}
\endpgfgraphicnamed
\hfill{}
\caption{Model}
\label{fig:model}
\end{figure}

We consider a system with $K$ users. Time is slotted. The system is operated in blocks of $B$ time slots. The length of a block, $B$, is a design parameter. Let $s(t)=1+(t-1)\bmod B$ and $n(t)=\lceil t/B \rceil$, \ie time slot $t$ is the $s(t)$-th time slot in block $n(t)$. Also, let $t(n,s)=s+(n-1)B$, \ie the $s$-th time slot of block $n$ is $t(n,s)$.

The random access feature of the model is captured by state variables $S_i$ which can be zero or one, depending on whether a user is active ($1$) or inactive ($0$). Let $S_i(n)$ denote the state of user $i$ in block $n$.  The state of a user is independent and identically distributed over all blocks and independent of the state of other users. Users are active with probability $a$, \ie $\Pr(S_i(n)=1)=a$ for all $i=1,\dots,K$.


Let $X_i[t]\in\RR$ and $Y[t]\in\RR$ denote the signal transmitted by user $i$, respectively the signal obtained by the receiver, in time slot $t$. We consider an AWGN channel without fading, \ie
\begin{equation} \label{eq:channel}
 Y[t] = \sum_{i=1}^K S_i[t] X_i[t] + Z[t],
\end{equation}
where $S_i[t]\triangleq S_i(n(t))$ and $\{Z[t]\}$ is white Gaussian noise with unit variance.

We consider transmission over $N$ blocks of length $B,$
and will denote the real-valued sequences of length $NB$ transmitted by
the $K$ transmitters and the real-valued sequence of length $NB$ received
at the destination by
\begin{align}
X_i &= \left( X_i[1], X_i[2], \cdots, X_i[NB]\right), \\
Y &= \left( Y[1], Y[2], \cdots, Y[NB]\right),
\end{align}
respectively.

In the description of our proposed strategy, we will also find it convenient to
index individual elements in each block. To this end, 
using the shorthand
$X_i(n,s)=X_i[t(n,s)]$, $Y(n,s)=Y[t(n,s)]$ and $Z(n,s)=Z[t(n,s)],$
we introduce the notation
\begin{align}
 X_i(n) &= \left( X_i(n,1), X_i(n,2), \cdots, X_i(n,B)\right), \\
Y(n) &= \left( Y(n,1), Y(n,2), \cdots, Y(n,B)\right).
\end{align}
The channel model, \eqref{eq:channel}, can alternatively be written as
\begin{equation}
 Y(n,s) = \sum_{i=1}^K S_i(n) X_i(n,s) + Z(n,s).
\end{equation}
The model is illustrated for two users, \ie for $K=2$, in Figure~\ref{fig:model}.

\begin{definition}[Strategy]
A strategy defines encoders $E_i$, $i=1,\dots,K$, that map the user message $M_i$ and channel states $S_i$ to a signal $X_i$, \ie 
\begin{equation}
 E_i:  \left\{1,\dots,2^{NBR_i}\right\} \times\{0,1\}^{N}\to\RR^{NB},
\end{equation}
where we require these mappings to satisfy the following average power constraint:
\begin{equation}
 \frac{1}{NB}\sum_{t=1}^{NB} {\mathbb E} \left[X_i^2[t] \right] \leq P_i,
\end{equation}
for all $i=1,\dots,K$,
where the expectation is over all messages.
We denote the encoder mapping by $X_i=E_i(M_i,S_i)$.
Finally, a strategy also defines a decoder $D$ that uses knowledge of user states to map the received signal to an estimate of the user messages, \ie
\begin{equation}
  D: \RR^{NB}\times\{0,1\}^{KN}  \to \left\{1,\dots,2^{NBR_1}\right\}  \times \cdots \times \left\{1,\dots,2^{NBR_K}\right\},
\end{equation}
We denote the decoder mapping by
$\hat M \triangleq (\hat M_1,\dots,\hat M_K) = D(Y,S_1,\dots,S_K)$.
\end{definition}

For a given strategy, we define the resulting average error probability by
\begin{equation}
 \Pe =  \Pr\left[ (\hat M_1,\dots,\hat M_K)   \neq ( M_1,\dots, M_K)  \right].
\end{equation}

In the present paper, we restrict attention to the {\em symmetric } scenario where
all rates and all power constraints are equal, i.e.,
\begin{eqnarray}
R = R_1 = \ldots = R_K & \mbox{ and } & P = P_1 = \ldots = P_K, \nonumber
\end{eqnarray}
and we will refer to the {\em throughput} $T$ of a strategy simply as the sum of all the user rates, 
namely,
\begin{equation} \label{eq:T}
T=KR.
\end{equation}

The goal of the present paper is to characterize {\em achievable throughput,} which we define in
a standard information-theoretic manner:
\begin{definition}[Achievable throughput]
Throughput $T$ is achievable if there exists for every $\epsilon_1>0$ and $\epsilon_2>0$ a strategy with $R=T/K-\epsilon_1$ for which $\Pe\leq\epsilon_2$.
\end{definition}

{\em Remark.}
A few remarks about our modeling assumptions are in order:
\begin{enumerate}
\item Users are assumed to know $K$, the total number of users in the system.
\item There is block synchronization between users.
\item There is no feedback from the receiver to the users.
\item The long-term average power constraint allows to perform power control: Transmit at power $a^{-1}P$ in a block in which a user is active and with zero power otherwise, leading to long-term average power $P$.
\end{enumerate}
All of these could be relaxed, and we will discuss some of these relaxations in the discussion at the end.


\begin{figure*}
\centering
\begin{tikzpicture}[scale=.6]
\tikzstyle{H}=[draw, rectangle,minimum width=10mm,minimum height=10mm];
\tikzstyle{boxb}=[draw, circle,inner sep=1pt];
\tikzstyle{->}=[-latex];


\draw[decoration={brace,amplitude=1mm},decorate] (-4.8,6) -- (-1.2,6);
\draw[decoration={brace,amplitude=1mm},decorate] (1.2,6) -- (4.8,6);

\node at (-3,7) {\fm{M_i(1)}};
\node at (3,7) {\fm{M_i(2)}};

\begin{scope}[xshift=-5cm,yshift=5cm]
\draw (0,-.5) rectangle (4,.5);
\foreach \x in {1, 2, 3}
 {
    \draw (\x,-.5) -- (\x,.5);
}
\end{scope}

\begin{scope}[xshift=1cm,yshift=5cm]
\draw (0,-.5) rectangle (4,.5);
\foreach \x in {1, 2, 3}
 {
    \draw (\x,-.5) -- (\x,.5);
}
\end{scope}

\node[boxb]  (c11) at (-6.5,3) {\fm{\times}};
\node[boxb]  (c21) at (-.5,3) {\fm{\times}};
\node[boxb]  (c31) at (5.5,3) {\fm{\times}};
\node[boxb]  (c12) at (-5.5,3) {\fm{\times}};
\node[boxb]  (c22) at (.5,3) {\fm{\times}};
\node[boxb]  (c32) at (6.5,3) {\fm{\times}};

\node[inner sep=0] (g11) at (-8,2.5) {\fm{g_i(1,1)}};
\node[inner sep=0] (g12) at (-4,2.5) {\fm{g_i(1,2)}};
\node[inner sep=0] (g21) at (-2,2.5) {\fm{g_i(2,1)}};
\node[inner sep=0] (g22) at (2,2.5) {\fm{g_i(2,2)}};
\node[inner sep=0] (g31) at (4,2.5) {\fm{g_i(3,1)}};
\node[inner sep=0] (g32) at (8,2.5) {\fm{g_i(3,2)}};

\draw[->] (g11) -- (c11);
\draw[->] (g12) -- (c12);
\draw[->] (g21) -- (c21);
\draw[->] (g22) -- (c22);
\draw[->] (g31) -- (c31);
\draw[->] (g32) -- (c32);

\draw[->] (-3,4.5) -- (c11);
\draw[->] (3,4.5) -- (c12);
\draw[->] (-3,4.5) -- (c21);
\draw[->] (3,4.5) -- (c22);
\draw[->] (-3,4.5) -- (c31);
\draw[->] (3,4.5) -- (c32);

\node[boxb]  (pl11) at (-6,2) {\fm{+}};
\node[boxb]  (pl21) at (0,2) {\fm{+}};
\node[boxb]  (pl31) at (6,2) {\fm{+}};

\draw[->] (c11) -- (pl11);
\draw[->] (c12) -- (pl11);
\draw[->] (c21) -- (pl21);
\draw[->] (c22) -- (pl21);
\draw[->] (c31) -- (pl31);
\draw[->] (c32) -- (pl31);

\draw[->] (pl11) -- (-6,1);
\draw[->] (pl21) -- (0,1);
\draw[->] (pl31) -- (6,1);

\begin{scope}[xshift=-8cm,yshift=.5cm]
\draw (0,-.5) rectangle (4,.5);
\foreach \x in {1, 2, 3}
 {
    \draw (\x,-.5) -- (\x,.5);
}
\node[fill=white,inner sep=1pt] at (2,0) {\fm{\tilde M_i(1)}};
\end{scope}

\begin{scope}[xshift=-2cm,yshift=.5cm]
\draw (0,-.5) rectangle (4,.5);
\foreach \x in {1, 2, 3}
 {
    \draw (\x,-.5) -- (\x,.5);
}
\node[fill=white,inner sep=1pt] at (2,0) {\fm{\tilde M_i(2)}};
\end{scope}

\begin{scope}[xshift=4cm,yshift=.5cm]
\draw (0,-.5) rectangle (4,.5);
\foreach \x in {1, 2, 3}
 {
    \draw (\x,-.5) -- (\x,.5);
}
\node[fill=white,inner sep=1pt] at (2,0) {\fm{\tilde M_i(3)}};
\end{scope}

\node[H] (h1) at (-6,-2) {\fm{F}};
\node[H] (h2) at (0,-2) {\fm{F}};
\node[H] (h3) at (6,-2) {\fm{F}};

\draw[->] (-6,0) to (h1);
\draw[->] (0,0) to (h2);
\draw[->] (6,0) to (h3);

\node[boxb] (m11) at (-6,-4) {\fm{\times}};
\node[boxb] (m21) at (0,-4) {\fm{\times}};
\node[boxb] (m31) at (6,-4) {\fm{\times}};

\draw[->] (h1) to (m11);
\draw[->] (h2) to (m21);
\draw[->] (h3) to (m31);

\draw[->] (m11) to (-6,-5.5);
\draw[->] (m21) to (0,-5.5);
\draw[->] (m31) to (6,-5.5);

\node (s1) at (-8,-4) {\fm{a^{-1}S_i(1)}};
\node (s2) at (-2,-4) {\fm{a^{-1}S_i(2)}};
\node (s3) at (4,-4) {\fm{a^{-1}S_i(3)}};

\draw[->] (s1) -- (m11); 
\draw[->] (s2) -- (m21); 
\draw[->] (s3) -- (m31);

\begin{scope}[xshift=-9cm,yshift=-6cm]
\node at (-1,0) {\fm{X_i}};
\draw[thick] (0,-.5) rectangle (18,.5);
\foreach \x in {1, 2, ..., 17}
 {
    \draw (\x,-.5) -- (\x,.5);
}
\draw[ultra thick] (6,-.5) -- (6,.5);
\draw[ultra thick] (12,-.5) -- (12,.5);
\end{scope}

\draw[decoration={brace,amplitude=1mm},decorate] (-3.2,-7) -- (-8.8,-7);
\draw[decoration={brace,amplitude=1mm},decorate] (2.8,-7) -- (-2.8,-7);
\draw[decoration={brace,amplitude=1mm},decorate] (8.8,-7) -- (3.2,-7);

\node at (-6,-8) {\fm{X_i(1)}};
\node at (0,-8) {\fm{X_i(2)}};
\node at (6,-8) {\fm{X_i(3)}};

\end{tikzpicture}
\caption{The encoder for user $i$.\label{fig:support}}
\end{figure*}
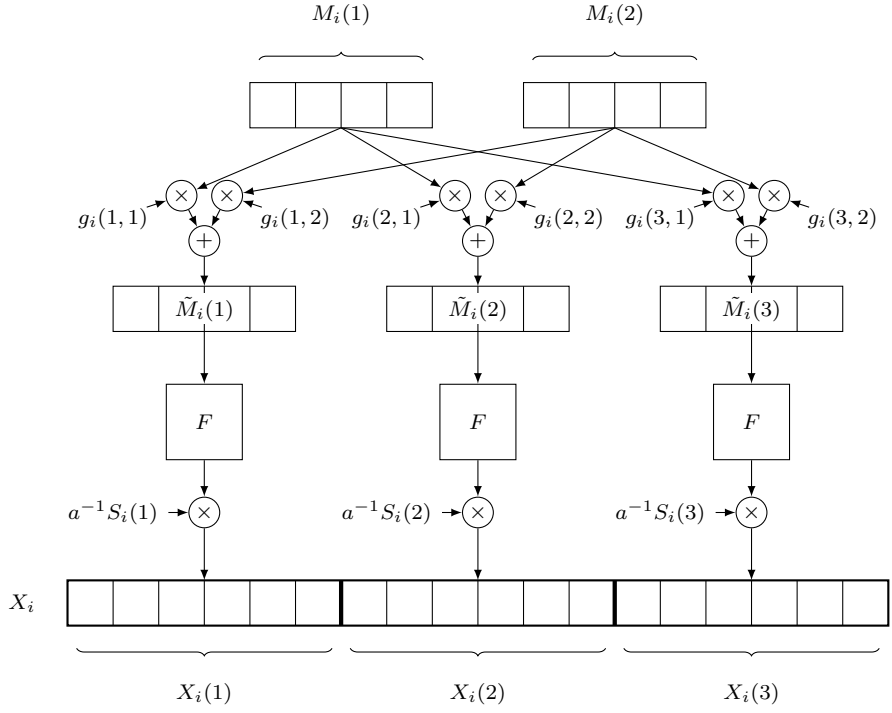

%
%
%
\section{Proposed Strategy} \label{sec:strategy}

The strategy presented in this section forms the main contribution of the current paper.


\subsection{Message structure}
Remember from Section~\ref{sec:model} that user $i$ has a message $M_i$ to transmit
\begin{equation}
 M_i\in\left\{1,\dots,2^{NBR}\right\}.
\end{equation}
The first step of the proposed strategy consists is expressing the message
$M_i$ as a string of $NBR/\log_2q$ symbols from $\FF_q,$
where $q$ will be suitably chosen.
These symbols are grouped in $L_b$ message substrings $M_i(\ell),$ each of length $L_c,$ such that we can express
\begin{equation}
 M_i = \left( M_i(1), \dots M_i(L_b)\right),
\end{equation}
where $M_i(\ell)\in\FF_q^{L_c},$ with
\begin{equation}
L_c = \frac{NBR}{L_b\log_2 q}.  \label{Eq-strategy-rate}
\end{equation}
The values of $L_b$ and $L_c$ need to be chosen carefully; an analysis is provided in subsection~\ref{ssec:peva}.

\subsection{Encoder}
The encoder $E_i$ at user $i$ consists of: 
\begin{itemize}
\item Matrix $\GG_i=[g_i(n,\ell)]$ of size $N\times L_b$ with elements from $\FF_q$,
\item Computation code $F:\FF_q^{L_c}\to\RR^B$, \ie a code of blocklength $B$ that takes $L_c$ message symbols. 
\end{itemize}

Encoder $i$ constructs the signal $X_i$, by performing the following steps for each block $n=1,\dots,N$:
\begin{enumerate}
%
%
\item The encoder first computes new equivalent message substrings $\tilde{M}_i(n)$ by mixing the original message substrings $M_i(\ell),$
as follows:
\begin{equation}
 \tilde{M}_i(n) = \sum_{\ell=1}^{L_b} g_i(n,\ell) M_i(\ell),
\end{equation}
where all operations on $M_i(\ell)$ are componentwise.
\item Use $\tilde{M}_i(n)$ as the input of computation code $F$ (as in Eqn.~\eqref{Eq-compcodeH}) and take the user state into consideration, \ie 
\begin{equation}
 X_i(n) = a^{-1}S_i(n)F(\tilde{M}_i(n)).
\end{equation}
\end{enumerate}

The encoder strategy is illustrated in Figure~\ref{fig:support}.

\subsection{Decoder}
The receiver decodes as follows:
\begin{enumerate}
\item In block $n$ the receiver observes the signal
\begin{equation}
 Y(n) = a^{-1}\sum_{\substack{i\in \{1,\dots,K\}:\\ \ S_i(n)=1}} F(\tilde{M}_i(n)) + Z(n).
\end{equation}
\item Assuming that $L_c$ is chosen properly such that the computation rate is achievable,
the computation code thus enables the decoder to recover
\begin{align} \label{eq:b}
 b(n) &= \sum_{\substack{i\in \{1,\dots,K\}:\\ \ S_i(n)=1}} \tilde{M}_i(n) \notag \\
&= \sum_{\substack{i\in \{1,\dots,K\}:\\ \ S_i(n)=1}}\ \sum_{\ell=1}^{L_b} g_i(n,\ell) M_i(\ell), 
\end{align}
for all $n=1,\dots,N$.
\item It remains to retrieve the messages by solving the system of linear equations given by~\eqref{eq:b}. It is shown in the next subsection that the right choice of $L_b$ results in a full rank system of equations. 
\end{enumerate}

\subsection{Performance analysis} \label{ssec:peva}

First of all, we consider the use of the computation code, which in our strategy is done separately in each block
of $B$ channel uses, and involves strings of $L_c$ symbols from $\FF_q.$
By combining Theorem~\ref{th:compfw} with Eqn. \eqref{Eqn-compfw-rate},
we know that if we choose
\begin{align}
  \frac{L_c \log_2 q }{B} &= \frac{1}{2}\log_2\left(\frac{1}{K}+a^{-1}P\right), \label{eq:Lc}
\end{align}
we can make the error probability arbitrarily small if $B$ is chosen sufficiently large.

Moreover, we have to choose $L_b$ such that the system of linear equations (over $\FF_q$) given by~\eqref{eq:b}
is full rank. We demonstrate that by choosing
\begin{align} \label{eq:Lb}
  \frac{L_b}{N} &= \frac{1-(1-a)^K}{K},
\end{align}
we can make this probability arbitrarily close to one if $N$ is chosen sufficiently large.

By~\eqref{eq:T} and~\eqref{Eq-strategy-rate}, we have that
\begin{eqnarray}
  T  =  K R & = &  \frac{K L_b L_c \log_2 q}{NB}.
\end{eqnarray}
The above choices of $L_b$ and $L_c$ thus lead to the following theorem, which forms the main contribution of the current paper.
\begin{theorem} \label{th:main}
Throughput
\begin{equation}
 T = \left(1- \left(1-a\right)^K\right)\frac{1}{2}\log\left[\frac{1}{K}+a^{-1}P\right]
\end{equation}
is achievable using a physical-layer network coding strategy.
\end{theorem}
A full proof is omitted. It consists mostly of demonstrating that the system of linear equations given by~\eqref{eq:b} is full rank. In the remainder of this section we provide the intuition.

We rewrite~\eqref{eq:b} as
\begin{equation} \label{eq:system}
 \AA\MM=\bb,
\end{equation}
where
\begin{equation}
\MM =  \begin{bmatrix}
  \MM_1 \\  \vdots \\ \MM_K
\end{bmatrix},
\quad \MM_i=
\begin{bmatrix}
  M_i(1) \\ \vdots \\  M_i(L_b)     
\end{bmatrix},\quad
\bb =
\begin{bmatrix}
  b(1) \\ \vdots \\  b(L_b)     
\end{bmatrix}
\end{equation}
and
\begin{equation}
 \AA = \SS\GG,
\end{equation}
with $\SS=[\SS_1 \cdots \SS_K]$ and $\GG=\diag(\GG_1,\dots,\GG_K)$.

Consider again the example from Figure~\ref{fig:introplnc}, represented in~\eqref{eq:intromatrix}. In this case $\AA$ can be written as
\begin{equation}
 \AA =
\begin{bmatrix}
 \gg_1(1) & \zz \\
 \gg_1(2) & \gg_2(2) \\
 \zz & \gg_2(3) \\
 \zz & \zz \\
 \gg_1(5) & \zz
\end{bmatrix},
\end{equation}
where $\gg_i(n)$ denotes the $n$-th row of $\GG_i$. It can be proven that the probability that the zeros in the matrix $\AA$ cause a rank defficiency is vanishing for large $N$.

\section{Comparison to Other Strategies} \label{sec:comparison}

In this section we present results on the achievable throughput of various other strategies. Subsection~\ref{ssec:aloha} deals with slotted ALOHA. The throughput of multipacket reception is analyzed in Subsection~\ref{ssec:minero}. A strategy that ignores that state of the user is presented in~\ref{ssec:ignore}. Finally, an upper bound to the achievable throughput is given in Subsection~\ref{ssec:upper}. A numerical evaluation of all these results is given in the next section.

\subsection{Slotted ALOHA} \label{ssec:aloha}
In slotted ALOHA users transmit at the maximum achievable rate, where the maximum is under the condition that there is only a single user in the system, \ie 
\begin{equation}
 R = \frac{1}{2}\log_2(1+a^{-1}P).
\end{equation}
The receiver is only able to decode if there is a single active user.

\begin{theorem}
Slotted ALOHA achieves throughput
\begin{equation}
 T = Ka(1-a)^{K-1}\frac{1}{2}\log_2\left(1+a^{-1}P\right).
\end{equation}
\end{theorem}
\begin{proof}
The probability that any particular user is the only user is $a(1-a)^{K-1}$.
\end{proof}

\subsection{Multipacket reception} \label{ssec:minero}
In this section we consider multipacket reception together with adaptive rates~\cite{minero2012random,medard2004capacity,minero2009mpr}. The strategy consists of choosing the rate such that the packets of all active users can be decoded as long as this number is at most $\tilde K$, \ie all users transmit at rate
\begin{equation}
 R=\frac{1}{2\tilde K}\log_2\left[1+\tilde Ka^{-1}P\right].
\end{equation}
Part of the strategy is to optimize over $\tilde K$.

\begin{theorem}[\cite{minero2012random}]
Multipacket reception achieves throughput
\begin{equation}
 T = \max_{\tilde K} \sum_{i=1}^{\tilde K} \binom{K}{i}a^i(1-a)^{K-i} \frac{i}{2\tilde K}\log_2\left[1+\tilde Ka^{-1}P\right].
\end{equation}
\end{theorem}


\subsection{Ergodic limit while ignoring CSI} \label{ssec:ignore}
If we ignore CSI we have an ordinary fading channel for which capacity is known~\cite{tse1998multiaccess}. In particular, this capacity can be achieved by coding across many blocks. In this case the users can not use power $a^{-1}P$ in active blocks, but are restricted to using average power $P$ in all blocks. For two users we have
\begin{equation}
 T = a(1-a)\log[1 + P] + \frac{a^2}{2}\log_2[1 + 2P].
\end{equation}
The general case is treated next.
\begin{theorem}[\cite{tse1998multiaccess}]
A throughput
\begin{equation}
 T = \sum_{i=1}^K \binom{K}{i}a^i(1-a)^{K-i}\frac{1}{2}\log_2\left(1+i P\right).
\end{equation}
is achievable by ignoring CSI and coding across many blocks.
\end{theorem}



\begin{figure*}
\hfill
\begin{minipage}{.45\textwidth}
\includegraphics[width=\textwidth]{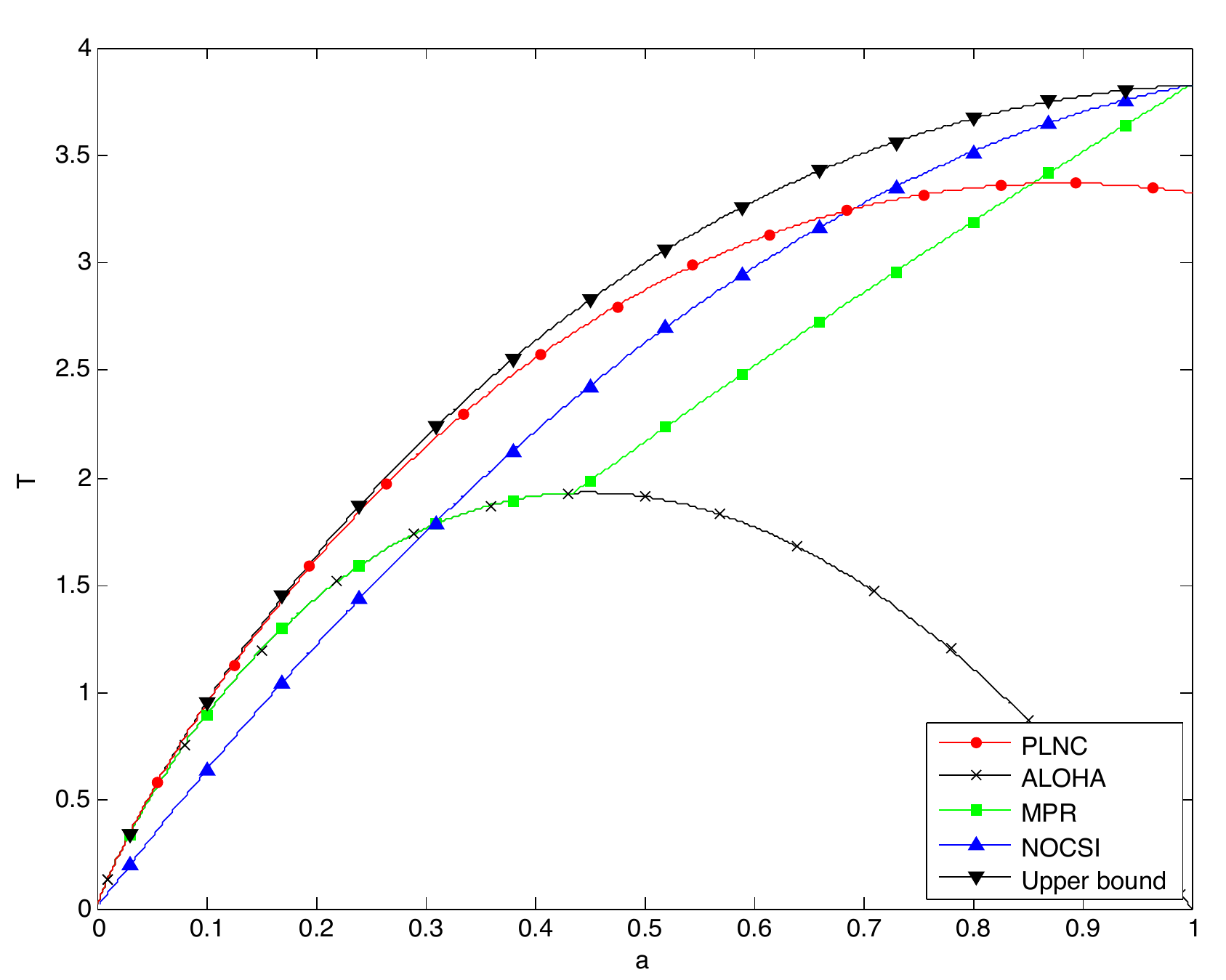}
\caption{Two users ($K=2$), low power ($P=10^2$). \label{fig:numAK2P2}}
\end{minipage}
\hfill
\begin{minipage}{.45\textwidth}
\includegraphics[width=\textwidth]{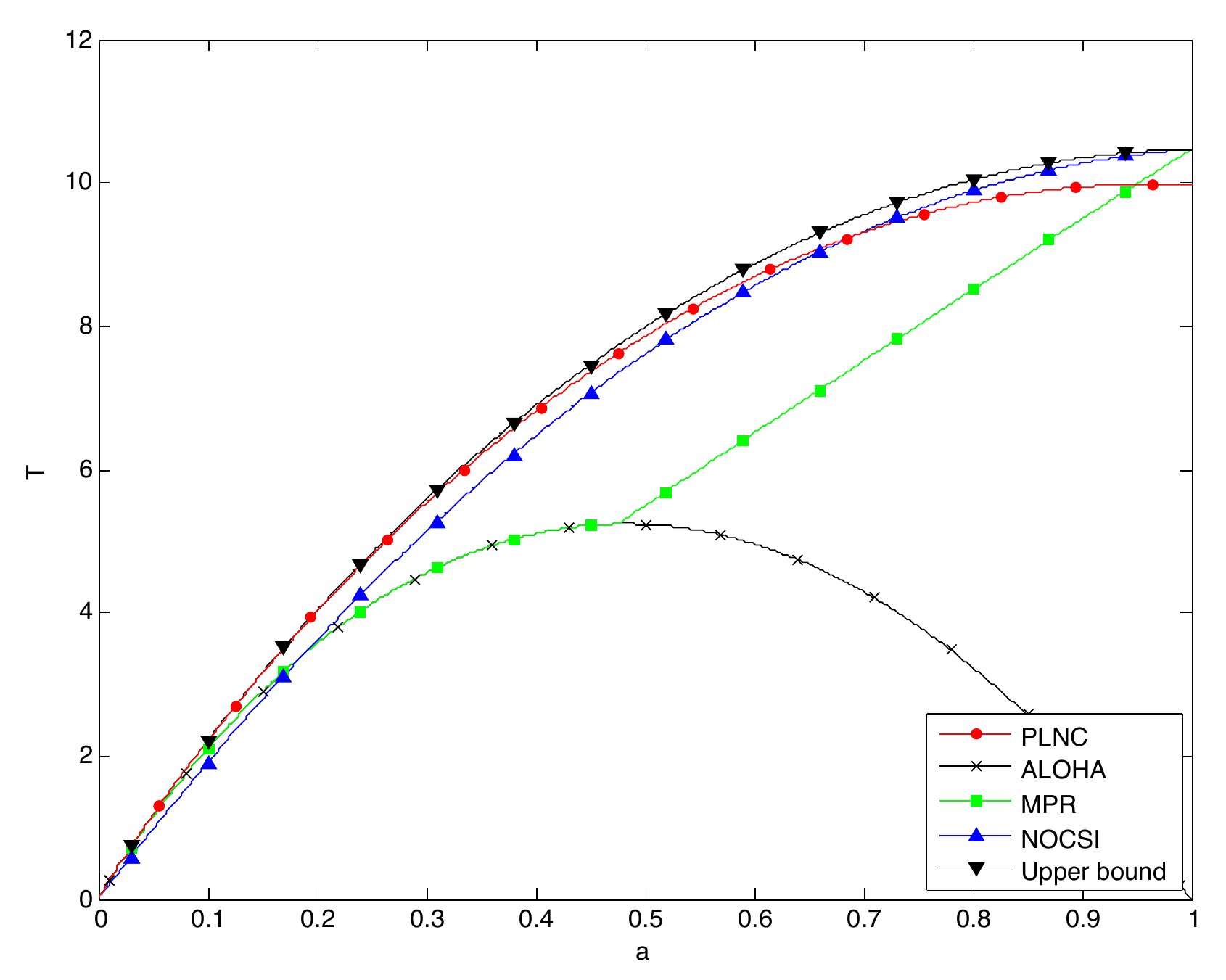}
\caption{Two users ($K=2$), high power ($P=10^6$). \label{fig:numAK2P6}}
\end{minipage}
\hfill{}

\vspace{6mm}

\hfill
\begin{minipage}{.45\textwidth}
\includegraphics[width=\textwidth]{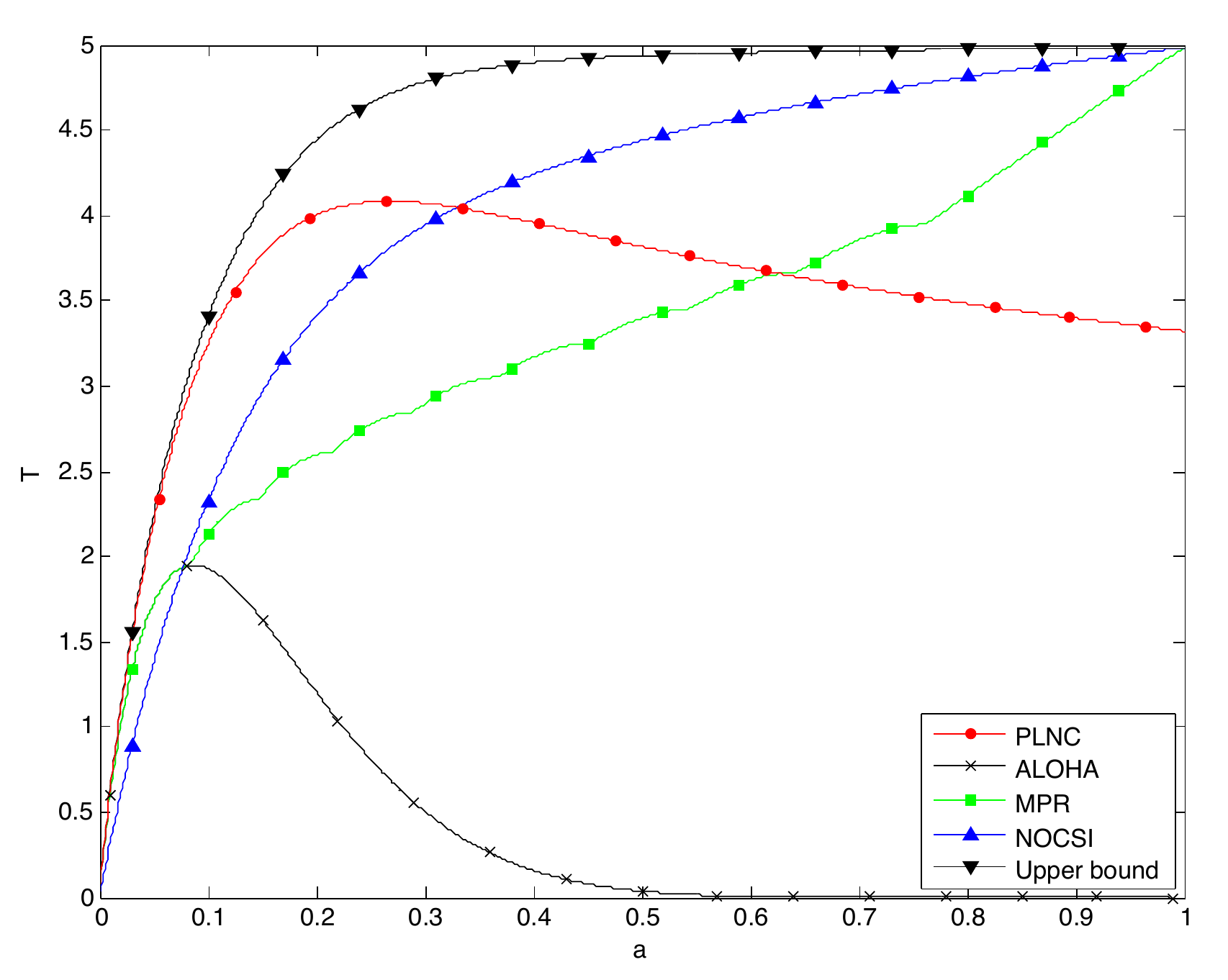}
\caption{Ten users ($K=10$), low power ($P=10^2$). \label{fig:numAK10P2}}
\end{minipage}
\hfill
\begin{minipage}{.45\textwidth}
\includegraphics[width=\textwidth]{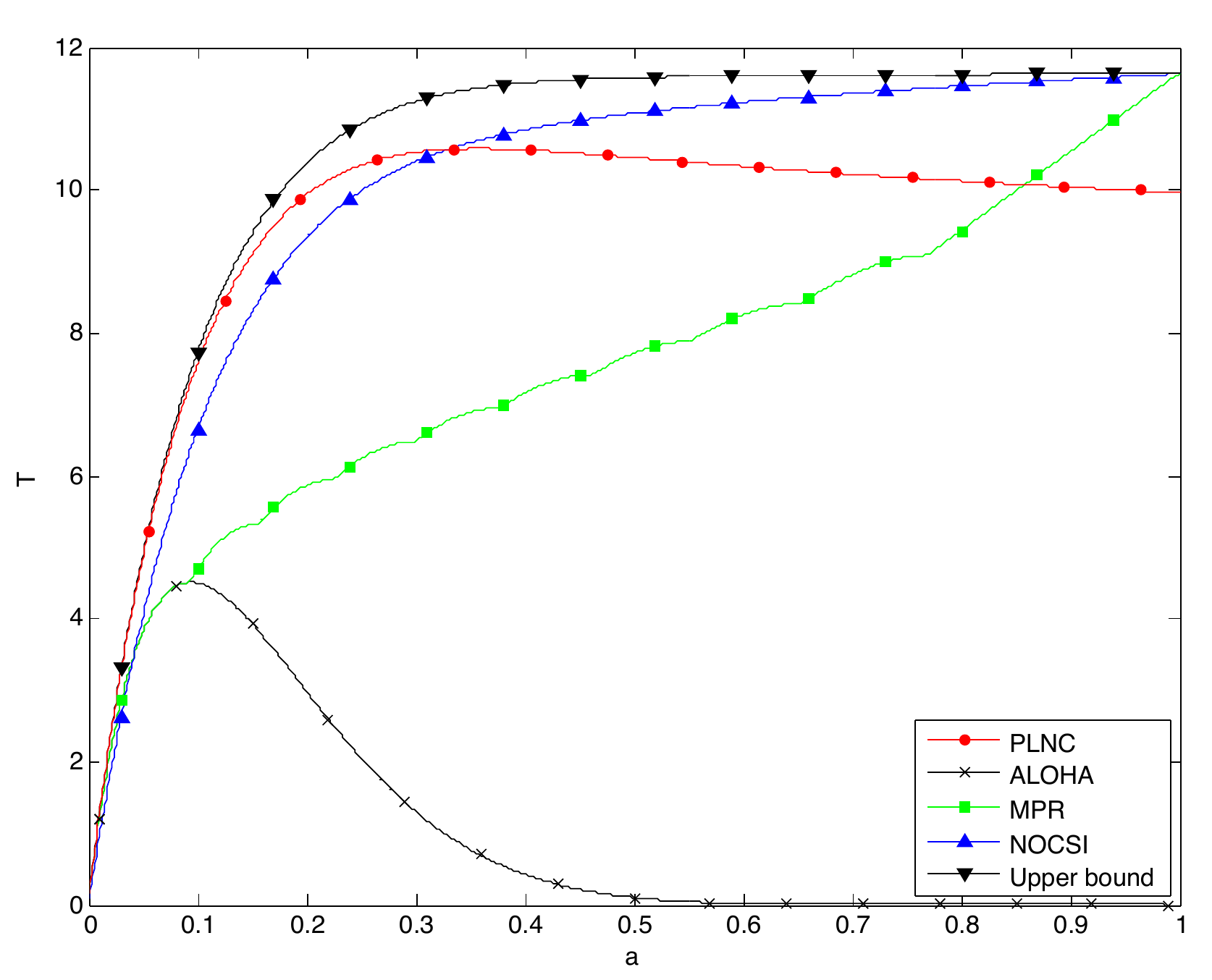}
\caption{Ten users ($K=10$), high power ($P=10^6$). \label{fig:numAK10P6}}
\end{minipage}
\hfill{}

\end{figure*}


\subsection{Upper bound} \label{ssec:upper}
We construct an upper bound by considering the case that users have complete knowledge of the state of all users.  A well known result, \cf~\cite{elgamalkim}, is that in this case the sum rate is upper bounded as
\begin{equation} \label{eq:upperstart}
 T \leq \EE\left[\frac{1}{2}\log_2\left(1+\sum_{i=1}^K S_i a^{-1} P\right)\right], 
\end{equation}
where the expectation is over the user states. For two users the bound reduces to
\begin{equation}
 T \leq a(1-a)\log_2\left[1+a^{-1}P\right] + \frac{a^2}{2}\log_2\left(1+2a^{-1}P\right).
\end{equation}
The general case follows readily by observing the number of active users is binomially distributed and is given in the next theorem.
\begin{theorem}
If throughput $T$ is achievable, then
\begin{equation}
 T \leq \sum_{i=1}^K \binom{K}{i}a^i(1-a)^{K-i}\frac{1}{2}\log_2\left(1+ia^{-1}P\right).
\end{equation}
\end{theorem}
\begin{IEEEproof}
Assume full CSI at all users.
\end{IEEEproof}


%
%
%
\section{Numerical evaluation} \label{sec:results}


\begin{figure*}[t]
\hfill
\begin{minipage}{.45\textwidth}
\vspace{8mm}
\includegraphics[width=\textwidth]{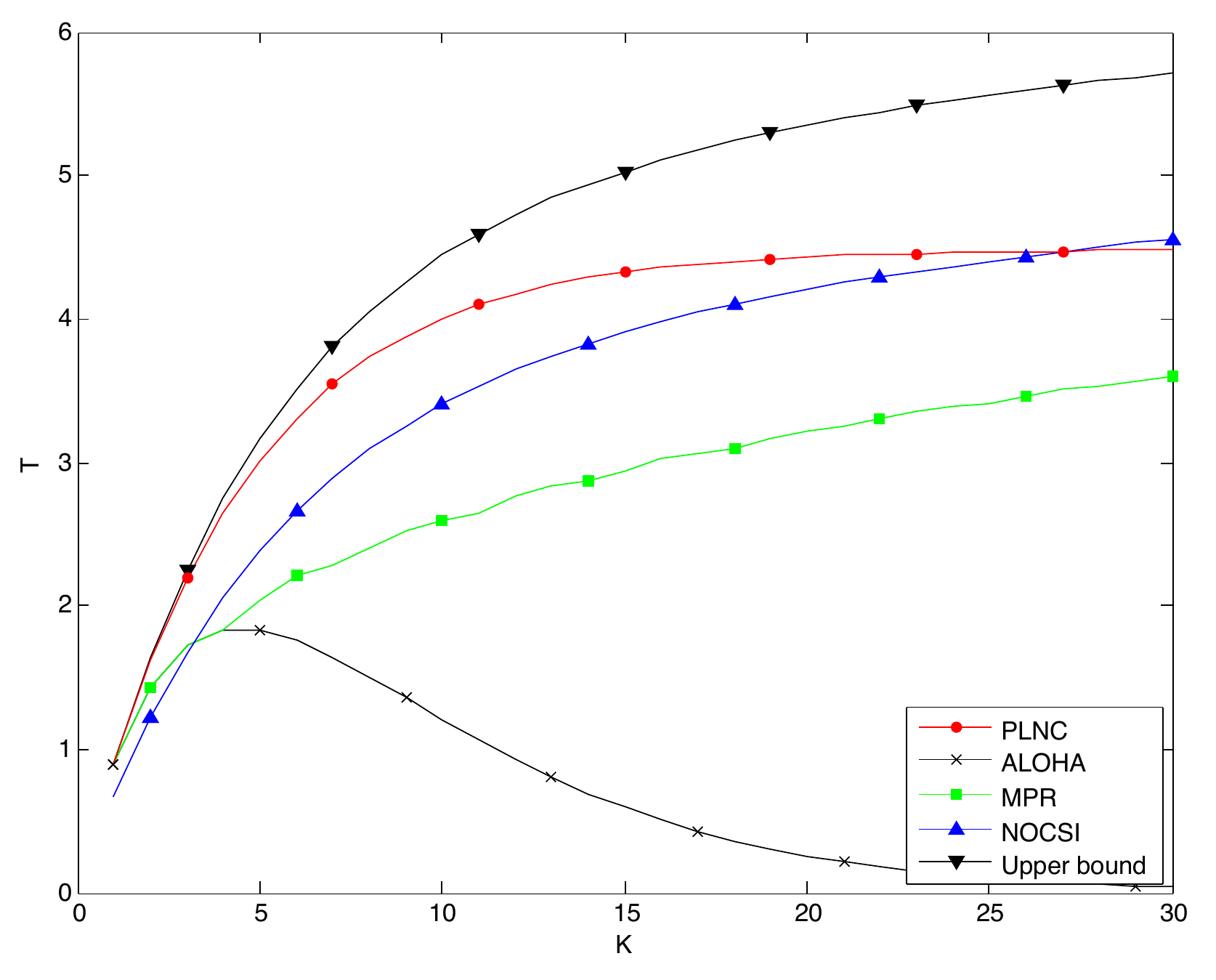}
\end{minipage}
\hfill
\begin{minipage}{.45\textwidth}
\vspace{8mm}
\includegraphics[width=\textwidth]{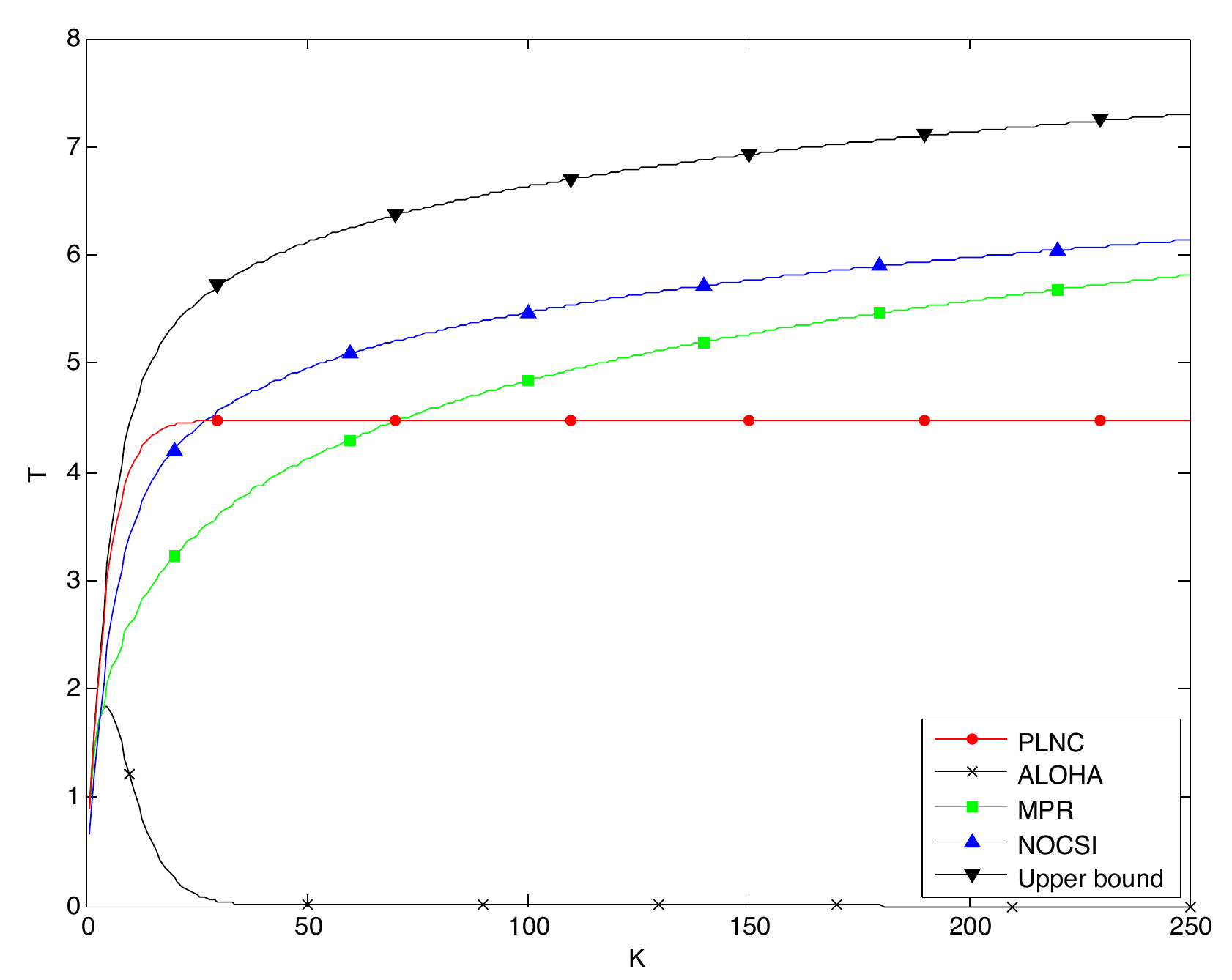}
\end{minipage}
\hfill{}
\caption{Low access probability ($a=0.2$), low power ($P=10^2$).\label{fig:numK}}
\end{figure*}

In this section we provide a numerical evaluation of the performance of the various strategies and the upper bound. In Figures~\ref{fig:numAK2P2}--\ref{fig:numAK10P6} we have plotted the throughput $T$ as a function of the access probability $a$ for various values for the number of users $K$ and the transmitter power $P$. In Figure~\ref{fig:numK} we have plotted $T$ as a function of $K$ for fixed $a$. 

The figures clearly demonstrate the well-known fact that ALOHA does not perform well for high access probability or many users. They also show the piecewise concave behavior of the throughput of the multipacket reception strategy. This behavior was  demonstrated analytically in~\cite{minero2012random}. For $a=1$, the model reduces to a classical multiple-access channel, \ie a channel without states. For such a channel multipacket reception is optimal. This is reflected in the figures, where at $a=1$, multipacket reception achieves the upper bound.

For moderate values of $a$ and $K$ the physical-layer network coding strategy performs significantly better than all other schemes. Moreover, the difference between the performance of the physical-layer network coding scheme and the upper bound is decreasing in the transmitter power.

Figure~\ref{fig:numK} suggests that the throughput of the physical-layer network coding strategy does not scale well in the number of users. It readily follows from Theorem~\ref{th:main} that the throughput of the strategy is less than $\frac{1}{2}\log_2(a^{-1}P)$, \ie it is constant in $K$. In contrast, the throughput of the multipacket reception scheme scales logarithmically in the number of users. However, this is not a fundamental limitation of the strategy introduced
here. Rather, it is an artifact of the fact that we are only considering
decoding a single linear combination of the packets. Instead, as the
number of users increases, we could decode {\em multiple } linear combinations
of the original packets, thus boosting the throughput.
Computation coding techniques are currently being extended to cover the
case of decoding multiple linear equations, with early work appearing in
\cite{nazer2012successive}.

We have compared the performance of our physical-layer network coding strategy with that of ALOHA and multipacket reception. Other strategies, that have not been taken into account in this analysis, have been proposed in the literature. In~\cite{minero2012random}, for instance, a superposition strategy is proposed in which information is organized in different layers, a subset of which is then opportunistically decoded. However, as oberved in~\cite{minero2012random}, on the AWGN channel it is optimal to use a single layer.

%
%


%
%
%
\section{Discussion} \label{sec:discussion}
We have presented an approach to random access that is based on physical-layer network coding. The gist of this strategy is that whenever packets collide, the receiver decodes a linear combination of these packets. The throughput that is achieved by this approach is significantly better than that of other approaches. 

The strategy as it is presented in the current paper does not employ feedback from the receiver. It can, however, easily be adjusted to incorporate feedback from the receiver. Consider, for instance, the case that the receiver provides an acknowledgement as soon as it decodes the system of linear equations~\eqref{eq:system} and recovers all message substrings from all users. The users can simply continue to transmit random linear combinations of their message substrings until they receive the acknowledgement. This results in a rateless strategy that is particularly easy to implement.

Observe that in our model we allow for strategies that perform coding over blocks. This potentially leads to very large delays. In fact, the proof of Theorem~\ref{th:main} is based on using many blocks. We believe that the strategy can be incorporated in a practical implementation, for which the number of blocks that is used is small. In particular, we can, in each block transmit a linear combination of only a small number of message substrings and use a sliding window. This will allow to receiver the decode substrings after a small number of blocks. Similar approaches are proposed for broadcast channels in~\cite{kumar08arq} and in~\cite{sundararajan2009tcp}.

Finally, the comparison with the arguments in~\cite{minero2012random} reveals a further
interesting insight. Namely, in~\cite{minero2012random}, an information-theoretic upper
bound is presented that applies to any strategy in which {\em packets}
are decoded after every block. It is clear that if we allow to code over
many blocks, the throughput can be enlarged beyond this upper bound,
but this requires the receiver to store all the physical-layer channel
outputs of all the blocks, which is a high price to pay in terms of
implementation complexity.

It is interesting to observe that the new strategy introduced in this paper
{\em also} permits to significantly outperform the upper bound presented
in~\cite{minero2012random}. This is not a contradiction since in our strategy, after every
block, we allow to decode linear combinations of packets, rather than
only entire packets. In this sense, the proposed strategy is a simple version
of coding over many blocks, where the receiver can discard the physical-layer
channel outputs after every block. Moreover, it should be expected that
a careful implementation of the strategy, will only require coding
over a small number of blocks.

%
%
%
\section{Acknowledgement}
This work was supported in part by the European ERC Starting Grant 259530-ComCom and by NWO grant $612.001.107$.


%
%
%

%
%
%
\bibliographystyle{IEEEtran}
\bibliography{IEEEabrv,randomaccess_arxiv}

\end{document}